\documentclass{article}
\usepackage{fullpage}
\usepackage[utf8]{inputenc}
\usepackage{amsmath, 
amsfonts, 
amsthm}
\usepackage{hyperref}
\usepackage{cleveref}
\usepackage{graphicx}
\usepackage{comment}
\usepackage{thm-restate}
\usepackage{bbm}
\usepackage{tikz}
\usepackage[linewidth=1pt]{mdframed}
\usepackage[most]{tcolorbox}
\definecolor{block-gray}{gray}{0.85}
\newtcolorbox{shadequote}{colback=block-gray,grow to right by=-2mm,grow to left by=-2mm,
boxrule=0pt,boxsep=0pt,breakable}

\usepackage{color}
\usepackage{enumitem}
\usepackage{cite}

\newtheorem{theorem}{Theorem}[section]
\newtheorem{lemma}[theorem]{Lemma}

\newtheorem{definition}{Definition}
\newtheorem{property}{Property}

\newcommand{\nummats}{n}

\newcommand{\recovery}{\textup{\texttt{Recovery}}}
\newcommand{\Recovery}{\textup{\texttt{Recovery}}}

\newmdtheoremenv{theo}{Theorem}
\begin{document}
\begin{titlepage}
\title{Near-Optimal Fault Tolerance
for Efficient Batch Matrix Multiplication
via an Additive Combinatorics Lens}

\author{
Keren Censor-Hillel\\
Technion\\ 
Israel\\
\textit{ckeren@cs.technion.ac.il}\\
	\and
Yuka Machino\\
MIT\\
USA \\ 
\textit{yukam997@mit.edu}\\ 
	\and
Pedro Soto\\
Oxford University \\
UK\\
\textit{Pedro.Soto@maths.ox.ac.uk}\\
}

\date{}

\maketitle

\begin{abstract}
Fault tolerance is a major concern in distributed computational settings. In the classic master-worker setting, a server (the master) needs to perform some heavy computation which it may distribute to $m$ other machines (workers) in order to speed up the time complexity. In this setting, it is crucial that the computation is made robust to failed workers, in order for the master to be able to retrieve the result of the joint computation despite failures. A prime complexity measure is thus the \emph{recovery threshold}, which is the number of workers that the master needs to wait for in order to derive the output. This is the counterpart to the number of failed workers that it can tolerate. 

In this paper, we address the fundamental and well-studied task of matrix multiplication. Specifically, our focus is on when the master needs to multiply a batch of $n$ pairs of matrices. Several coding techniques have been proven successful in reducing the recovery threshold for this task, and one approach that is also very efficient in terms of computation time is called \emph{Rook Codes}. The previously best known recovery threshold for batch matrix multiplication using Rook Codes is $O(n^{\log_2{3}})=O(n^{1.585})$.

Our main contribution is a lower bound proof that says that any Rook Code for batch matrix multiplication must have a recovery threshold that is at least $\omega(n)$. Notably, we employ techniques from Additive Combinatorics in order to prove this, which may be of further interest. Moreover, we show a Rook Code that achieves a recovery threshold of $n^{1+o(1)}$, establishing a near-optimal answer to the fault tolerance of this coding scheme.
\end{abstract}

\thispagestyle{empty}
\end{titlepage}

\newpage
\setcounter{page}{1}

\section{Introduction}
\label{sec:introduction}
The master-worker computing paradigm has been extensively studied due to its ability to process large data-sets whose processing time is too expensive for a single machine. To address this, the main machine (master) may split the computation among $m$ workers. Various tasks have been studied, such as matrix multiplication, gradient descent, tensors/multilinear computation, and more~\cite{LLPPR2018, yma17, dfhjcg2020,YAA2018, DBJMG2018,ylrksa19, jj21, sl20, sfsl22,tldk17, cgw21, jj21a, fc21, sma21, rk20}. 

A common issue in such settings is the need to cope with failures, so that faulty workers do not prohibit the completion of the computational task at hand. Indeed, faults are a major concern in many distributed and parallel settings, and coping with them has been studied for several decades~\cite{AttiyaWBook, LynchBook}.

In the master-worker setting, the main complexity measure that addresses the robustness of an algorithm to faults is its \emph{recovery threshold}, defined as follows.

~\\\textbf{Definition of Recovery Threshold~\cite{yma17}.}
The \emph{recovery threshold} of an algorithm in the master-worker setting is the number of workers the master must wait for in order to be able to compute its output.
~\\

A low recovery threshold means that the algorithm can tolerate a higher number of faults. For example, if the master splits the computation to $m$ disjoint pieces and sends one to each worker then it cannot tolerate even a single failure and needs to wait for all $m$ workers. If it splits the computation to $m/2$ pieces and sends each piece to 2 workers then it only needs to wait for $m-1$ machines, and thus can tolerate one failure. Notice that with 2 failures this approach does not work, as the 2 failed workers could be responsible for the same piece of computation. Indeed, such simple \emph{replication} approaches are inefficient in terms of their robustness.

An efficient way for handling such faults is to use coding techniques. Coding has been an extremely successful technique in various distributed computations (see, e.g., \cite{Censor-HillelCG22, AlonBEGH19, Censor-HillelGH19, HaeuplerWZ20, EfremenkoKS20, AggarwalDHS19}) which mostly address communication noise). In particular, coding schemes have been successfully proposed for the task of \emph{matrix multiplication} in the master-worker setting, which is the focus of our work. Multiplying matrices is a vital computational task, which has been widely studied in distributed settings ~\cite{Drucker+PODC14, Censor-Hillel+DC19, LeGall16, Censor+TCS20, CHDKL20, GuptaHKSS22} and in parallel settings (see, e.g., ~\cite{AzadBBDGSTW16,MilakovicSNBB22,BallardDKS16} and many references therein).

Specifically, we address batch matrix multiplication in the master-worker setting, in which the master has $\nummats$ pairs of matrices, $\{(A_i,B_i)\}_{0\leq i\leq\nummats-1}$, whose $n$ products $\{A_i\cdot B_i\}_{0\leq i\leq\nummats-1}$ it needs to compute. All matrices $A_i$ have the same dimensions, as well as all matrices $B_i$.

To see why coding is useful, consider the above na\"ive approach for coping with faults by replication as applied for batch matrix multiplication. For simplicity, suppose that $m$ is a multiple of $n$, i.e., $m = \lambda n$ for some integer $\lambda$. In the replication approach, the master sends each sub-task $A_i \cdot B_i$ to $\lambda$ workers. It is straightforward to see that the recovery threshold here is $\Recovery_\text{replication} (n) \geq (n-1)\lambda +1 = m-\lambda+1$, as $\lambda$ failures could be of workers that are responsible for the same product, where for any given family $\mathcal{F}$ of algorithms for computing the $n$ products, we denote by $\Recovery_\mathcal{F}(n)$ the best recovery threshold of an algorithm in $\mathcal{F}$. In particular, this means that with simple replication, the recovery threshold is linear in the number $m$ of machines.

However, it is well-known that with coding it is possible to do better. Consider the case $n=2$. 
With replication and $m = 2 \lambda$ workers, we have that $\Recovery_\text{replication} (2) \geq \lambda +1$. In contrast, consider the fault tolerance provided by the following simple coding scheme, in which we define $ \tilde A (x) : = A_0 + A_1 x$ and $ \tilde B (x) : = B_0 + B_1 x$. The master sends to each worker $w$ the distinct values $\tilde A(x_w), \tilde B(x_w)$, and each worker $w$ then returns the product $\tilde A(x_w) \cdot \tilde B(x_w)$ to the master. For any 3 workers $u,v,w$, it holds that
$$
\begin{bmatrix}
    x_u^0 & x_u^1 & x_u^2 \\ 
       x_v^0 & x_v^1 & x_v^2 \\ 
          x_w^0 & x_w^1 & x_w^2 \\ 
\end{bmatrix}
\begin{bmatrix}
    A_0 B_0 \\ 
       A_0 B_1 + A_1 B_0 \\ 
         A_1 B_1  \\ 
\end{bmatrix} = 
\begin{bmatrix}
    \tilde A(x_u) \cdot \tilde B(x_u)\\ 
      \tilde A(x_v) \cdot \tilde B(x_v)\\ 
        \tilde A(x_w) \cdot \tilde B(x_w)  \\ 
\end{bmatrix},
$$
and since $ x_u \neq x_v \neq x_w$ then the first matrix is invertible. Thus, the master can retrieve $A_0\cdot B_0$ and $A_1\cdot B_1$ using the results of any 3 workers, giving that $\Recovery_\text{coded} (2)  \leq 3 $.
In particular, the coding scheme achieves a fault tolerance of $\frac{m-3}{m} =  1 - \frac{3}{m} $, which tends to $100\%$ as $m$ grows. However, the replication scheme can only achieve a fault tolerance of $ \frac{2\lambda - 
 (\lambda + 1)  }{2\lambda }  = \frac{ 
1  }{2} - \frac{1}{m} $, which is bounded by $50\%$ even as $m$ grows.

~\\\textbf{\underline{Rook Codes for batch matrix multiplication.}} We consider the Rook Codes given by the following form, proposed in~\cite{sl20, sfsl22}. Denote by
\begin{equation}\label{eq:code_form} 
\tilde A (x) = \sum_{i \in [\nummats]} A_i x^{p_i}, \ \tilde B (x) = \sum_{j \in [\nummats]} B_j x^{q_j} 
\end{equation}
two matrix polynomials that are generated by the master, for some sequences of non-negative integers $P=(p_0,\dots,p_{n-1})$ and $Q=(q_0,\dots,q_{n-1})$.
We denote:
$$L_{P,Q}:=|P+Q|.$$

For $0\leq w\leq m-1$, the master sends to worker $w$ the coded matrices $\tilde A (x_w)$ and $\tilde B (x_w)$, where $x_w\neq x_{w'}$ for every pair of different workers $w\neq w'$. Worker $w$ then multiplies the two coded matrices and sends their product back to the master. Note that $$\tilde{A}(x)\cdot\tilde{B}(x)=\left(\sum_{i\in[n]}A_{i}x^{p_i}\right)\left(\sum_{j\in[n]}B_{j}x^{q_j}\right)=\sum_{i,j\in[n]}A_{i}B_{j}x^{p_i+q_j}.$$ 

Note that while the degree of the above polynomial may be $\max\{P+Q\}$, the number of its non-zero coefficients is bounded by $L_{P,Q} =|P+Q|$, which may be significantly smaller. Thus, when the master receives back $L_{P,Q}$ point evaluations of $\tilde{A}(x)\cdot\tilde{B}(x)$, it can interpolate the polynomial. To be able to extract the coefficients $A_k\cdot B_k$ for all $k\in [n]$ out of the polynomial, we need $A_k\cdot B_k$ to be the only coefficient of $x^{p_k+q_k}$ in the polynomial. That is, the following property is necessary and sufficient.
\begin{property}[The decodability property]
\label{property:decodability}
For all $i,j,k \in [n]$, $p_k+q_k =p_i+q_j $ if and only if $i=k,j=k$.
\end{property}
Throughout the paper, we refer to codes given by Equation~\ref{eq:code_form} that satisfy Property~\ref{property:decodability} as Rook Codes for batch matrix multiplication in the master-worker setting.

To summarize, given two sets of integers $P,Q$ of size $|P|=|Q|=n$ for which Property~\ref{property:decodability} holds, we have that $\Recovery_{\text{Rook-Codes}}(n) = \min_{P,Q}\{L_{P,Q}\}$ (the minimum is taken here over all $P,Q$ that satisfy Property~\ref{property:decodability}). In order to derive the recovery threshold of Rook Codes for batch matrix multiplication, our goal is then to bound $L_{P,Q}$ for all such $P,Q$, or in other words, to bound $|P+Q|$.

~\\\textbf{\underline{How robust can Rook Codes be?}} A simple Rook Code can be derived from the Polynomial Codes of~\cite{yma17}, which are designed for multiplying a single pair of matrices. This would give 
$\tilde A (x) = \sum_{i \in [\nummats]} A_i x^{i}$ and $\tilde B (x) = \sum_{j \in [\nummats]} B_j x^{nj}$,
i.e., we have $p_i=i$ and $q_j=nj$ for every $i,j\in [n]$.
This would immediately bound the recovery threshold of Rook Codes by 
$\recovery_\text{Rook-Codes}(n) \leq ((n-1)+n(n-1))+1 = n^2$, which is already a significant improvement compared to the recovery threshold of replication since it does not depend on $m$ but rather only on the size of the input.
In~\cite{sfsl22}, an intuition was given for why the task is non-trivial, in the form of a lower bound that says that with $P = (0,1,2,...,n-1)$ (i.e., with $\tilde A (x) = \sum_{i \in [\nummats]} A_i x^{i}$), any choice of $Q$ has $L_{P,Q} \geq n^2/2$, which implies that in order to obtain a recovery threshold that is sub-quadratic in $n$, one cannot use such simple values for $P$.

\sloppy{
But for other values of $P,Q$, Rook Codes can indeed do better, as shown by~\cite{sfsl22}, who provided $P$ and $Q$ with $L_{P,Q} = O(n^{\log_2{3}})$, implying that $\recovery_{\text{Rook-Codes}}(\nummats) = O(n^{\log_2{3}})= O(n^{1.585})$. The parameters chosen for this are $P = Q = \{ v_0+ v_1 \cdot 3 + v_2\cdot 3^2 + ... + v_{\ell-1} \cdot 3^{\ell-1} ~|~ v_i  \in  \{0,1\} \}$, where $\ell=\log_2{n}$, assuming that $\ell$ is an integer (an assumption that can easily be removed by standard padding). One can verify that this choice yields a Rook Code (i.e., that Property~\ref{property:decodability} holds), and that the recovery threshold is $\Recovery_{\text{Rook-Codes}}(n) = O(3^{\ell}) = O(3^{\log_2{n}}) = O(2^{\log_2{3} \cdot \log_2{n}}) = O(n^{\log_2{3}})$.
}

This significant progress still leaves the following question open:

~\\
\noindent \fbox{%
    \parbox{\textwidth}{%
        \textbf{Question:} Are there Rook Codes for batch matrix multiplication in the master-worker setting with a recovery threshold that is linear in $n$?
    }%
}

\subsection{Our Contribution}
In this paper, we show that no Rook Code for batch matrix multiplication in the master-worker setting can get a strictly linear in $n$ recovery threshold, but that we can get very close to such value.
 
~\\\textbf{\underline{Lower Bound.}} Our main contribution is showing a super-linear lower bound on the recovery threshold of Rook Codes for batch matrix multiplication in the master-worker setting.

\begin{restatable}[Lower Bound]{theorem}{ThmLB}
\label{thm:x}
   Every Rook Code for batch matrix multiplication in the master-worker setting has a super-linear in $n$ recovery threshold. That is, 
    $$\Recovery_\text{Rook-Codes}(n) = \omega(n).$$
\end{restatable}

Our main tool for proving Theorem~\ref{thm:x} is the following theorem in additive combinatorics, whose proof is the key technical ingredient of our contribution (Section~\ref{sec:LB}).

\begin{restatable}{theorem}{ThmLBviaAC}
\label{theorem:lower}
Let $P=\{p_0,p_1,\dots, p_{n-1}\}$ and $Q=\{q_0,q_1,\dots, q_{n-1}\}$ be finite subsets of the integers, such that $0\in P$ and $0\in Q$. Suppose that for all $0\le i,j,k \le n-1$, $p_k+q_k =p_i+q_j $ if and only if $i=k,j=k$. Then $|P+Q|=\omega(n)$.
\end{restatable}

It is easy to see that Theorem~\ref{theorem:lower} is sufficient for proving Theorem~\ref{thm:x}, as follows.

\begin{proof}[\textbf{Proof of Theorem~\ref{thm:x} given Theorem~\ref{theorem:lower}}]
Let $P=\{p_0,p_1,\dots, p_{n-1}\}$ and $Q=\{q_0,q_1,\dots, q_{n-1}\}$ be sets of integers that define a Rook Code for batch matrix multiplication in the master-worker setting. In particular, Property~\ref{property:decodability} holds, and hence by Theorem~\ref{theorem:lower}, we have that $|P+Q|=\omega(n)$. This implies that $L_{P,Q} = |P+Q| = \omega(n)$ and thus $\Recovery_{\text{Rook-Codes}}(n) = \min_{P,Q}\{L_{P,Q}\} = \omega(n)$, as claimed.
\end{proof}

~\\\textbf{\underline{Upper Bound.}} Our second contribution is that there exist Rook Codes for batch matrix multiplication in the master-worker setting that asymptotically get very close to the bound in \Cref{thm:x}. We prove the following (Section~\ref{sec:UB}).

\begin{restatable}[Upper Bound]{theorem}{ThmUB}
\label{thm:y}
    There exist Rook Codes for batch matrix multiplication in the master-worker setting with a recovery threshold of $ne^{O(\sqrt{\log n})}$. That is, 
    $$\Recovery_\text{Rook-Codes}(n) = n^{1+o(1)}.$$
\end{restatable}

\subsection{Related Work}
\label{subsec:relatedWork}
The use of coding for matrix-vector and vector-vector multiplication in the master-worker setting was first established concurrently in~\cite{LLPPR2018, yma17, dfhjcg2020}. Shortly after, efficient constructions were given for general matrix-matrix multiplication~\cite{YAA2018, DBJMG2018}. 
In particular, \cite{YAA2018} established the equivalence of the recovery threshold for general partitions to the tensor rank of matrix multiplication.  Further, \cite{tldk17,sigl22} considered coding for gradient descent in a master-worker setting for machine learning applications. 
Batch matrix multiplication was studied using Rook Codes~\cite{sl20, sfsl22}, Lagrange Coded Computation (LCC) \cite{ylrksa19}, and Cross Subspace Alignment (CSA) codes \cite{jj21}.
Follow-up works about LCC consider numerically stable extensions \cite{fc21, sma21} and private polynomial computation \cite{rk20}. There have also been follow-up works on CSA codes which include extending them to obtain security guarantees \cite{cgw21,jj21a}.

It is important to mention that LCC and CSA codes have recovery thresholds of $2n-1$, and so, at a first glance, it may seem that LCC or CSA Codes already achieve a linear recovery threshold and thus outperform Rook Codes. However, an important aspect in which the Rook Codes that are considered in our work are powerful is in their computational complexity. 
The work in \cite{sfsl22} supplies experimental evidence that the encoding step for Rook Codes is computationally more efficient. In Section~\ref{appx:parr_and_comp}, we give detailed descriptions of LCC and LCA codes, and provide some mathematical intuition as to why this may be the case.

In \cite{oek20,oehk21}, additive combinatorics is used to construct and analyze coded matrix multiplication algorithms. 
However, there are two key differences between these works and ours. The first is that they consider the special case of outer product of two block matrix vectors, which corresponds to batch matrix multiplication but with dependencies between the pairs of matrices which can be exploited; in particular, 
their constraint is simpler than the decodability property we need (Property~\ref{property:decodability}).
The second difference is that they consider security of the computation. Hence, our results are incomparable.

\subsection{Roadmap} The following section shortly contains the notation for the additive combinatorics setup. Section~\ref{sec:LB} contains the proof of~\Cref{theorem:lower}, which we proved above to imply the lower bound of~\Cref{thm:x}. Section~\ref{sec:UB} proves the upper bound of~\Cref{thm:y}. We conclude in Section~\ref{appx:parr_and_comp} with a discussion of the computational efficiency that motivates our study of the recovery threshold of Rook Codes.

\section{Notation}
For two sets of integers $A$ and $B$, the set $A+B$ is defined as $$A+B := \{a+b ~|~ a\in A, b\in B\},$$ and the set $A-B$ is defined as $$A-B := \{a-b ~|~ a\in A, b\in B\}.$$ Note that when $A=B$ we still have the same definitions, so for example $A-A = \{a-b ~|~ a\in A, b\in A\}$, which is the set of all integers that are a subtraction of one element of $A$ from another. 

For a set of integers $A$ and an integer $k$, the set $kA$ is defined as $$kA := A + A + \dots + A = \{a_0+a_1+\dots+a_{k-1} ~|~ \forall 0\leq i \leq k-1, a_i \in A\}.$$

\section{An $\omega(\nummats)$ Lower Bound on the Recovery Threshold}
\label{sec:LB}

In this section we prove that no Rook Code for batch matrix multiplication can achieve a strictly linear recovery threshold. 
\ThmLB*

As proven in the introduction, the following theorem implies~\Cref{thm:x}, and the remainder of this section is dedicated to its proof. 

\ThmLBviaAC*

\Cref{theorem:lower} is closely related to Freiman's Theorem in~\cite{Freiman} (see also~\cite[Theorem 1.3, Chapter 2]{Ruzsa2009}), a deep result in Additive Combinatorics which essentially states that a set $P$ has a small \emph{doubling constant} (i.e., the size of $|P+P|$ is only a constant multiple of $|P|$) if and only if $P$ \emph{is similar to} an arithmetic progression, for a well defined notion of similarity. Our theorem is comparable to Freiman's Theorem in the following sense. The contrapositive version of \Cref{theorem:lower} states that if $|P+Q|$ is a constant multiple of $|P|$, then $P+Q$ \emph{looks like} an arithmetic progression (in a different way compared to Freiman's similarity notion), in the sense that there exists $i,j,k$ with $(i,j)\neq (k,k)$ such that $p_k+q_k = p_i+q_j$ (note that when $P=Q$, this implies that $P$ contains a $3$-term arithmetic progression, or $3$-AP for short). In order to prove Theorem~\ref{theorem:lower}, we will invoke the following theorems from Additive Combinatorics.

\begin{theorem}[Ruzsa's Triangle Inequality {\cite[Lemma 2.6]{tao_vu_2006}}]
\label{thm:ruzsa}
If $A,B,C$ are finite subsets of integers, then
    $$|A||B-C|\leq |A-B||A-C|.$$
\end{theorem}

\begin{theorem}[The Plünnecke-Ruzsa Inequality {\cite[Theorem 1.2]{Petridis_2012}}]
\label{thm:PR}
Let $A$ and $B$ be finite subsets of integers and let $K$ be a constant satisfying the following inequality: 
$$|A+B|\le K|A|.$$ 
Then for all integers $r,\ell\ge 0$,
$$|rB-\ell B|\leq K^{r+\ell}|A|.$$
\end{theorem}

We will also need the following lemma.
\begin{lemma}
\label{lemma:size_bound}
Let $P,Q\subseteq \mathbb{Z}$ such that $|P|=|Q|=n$ and $|P+ Q|\le K|P|=Kn$ for some constant $K$. Then $$|P+ Q- P- Q|\le K^7 n.$$
\end{lemma}
\begin{proof}
First, using Ruzsa's Triangle Inequality from \Cref{thm:ruzsa} with $A=P,B=P-P,$ and $ C=Q-Q$, we get that 
\begin{equation}\label{eq_first_ruzsa_triangle}
    |P+ Q-P-Q|\le |P-P+P||P+Q-Q|/|P|.    
\end{equation}

Next, using the Plünnecke-Ruzsa Inequality from \Cref{thm:PR}, where $A=Q$, $B=P$, $r=2,\ell =1$ and the condition $|P+Q|\le K|Q|$ (recall that $|P|=|Q|$), we get that 
\begin{equation}\label{eq_bound_set}
    |P-P+P|=|2P-P|\le K^{2+1}|Q|=K^3|P|.
\end{equation}
Using symmetry of $P$ and $Q$, we can similarly prove that 
$$|Q-Q+Q|\le K^3|Q|.$$
Hence, combining Equation (\ref{eq_first_ruzsa_triangle}) and Equation (\ref{eq_bound_set}) together, we get that $|P+ Q-P-Q|$ is bounded by
$$|P+ Q-P-Q|\le |P-P+P||P+Q-Q|/|P|\le K^{3}|P||P+Q-Q|/|P|=K^{3}|P+Q-Q|.$$
Using Ruzsa's Triangle Inequality again, this time with $A=Q, B=-P$ and $C=Q-Q$, and then using the fact that $|Q+Q-Q|\le K^3 |Q|$, we get that:
$$|P+Q-Q|=|-P+Q-Q|\le |P+Q||Q+Q-Q|/|Q|\le |P+Q|K^3 \le K^{4}n.$$

Plugging this into the bound on $|P+ Q-P-Q|$ gives that 
$|P+Q-P-Q|\le K^3 |P+Q-Q|\le K^7 n$, as claimed.
\end{proof}

Now we are ready to prove Theorem~\ref{theorem:lower}, for which we will also use the following result:

\begin{theorem}[The Triangle Removal Lemma {\cite[Lemma 10.46]{tao_vu_2006}}]
\label{theorem:triangleRemoval}
Let $G=G(V,E)$ be a graph which contains at most $o(|V|^3)$ triangles. Then it is possible to remove $o(|V|^2)$ edges from $G$ to obtain a graph which is triangle-free.
\end{theorem}

\textbf{Overview:}
First we will assume for contradiction that there exists some constant $K$ such that for all $n$, there exists $P,Q$ with $|P|=|Q|=n$ such that $|P+Q|\le Kn$ and $P,Q$ satisfy Property~\ref{property:decodability}. From this assumption, we will construct a tripartite graph $G=(V,E)$ with $V=(X,Y,Z)$ such that every edge is in exactly one triangle in $G$. Each of the three vertex sets will have size $N=|P+Q-P-Q|$ and the number of edges between each pair of sets is at least $Nn/4$. Because every edge is in exactly one triangle, we have that the number of triangles is exactly $|E|/3$, and thus it is $o(|V|^3)$. By using the Triangle Removal Lemma from \Cref{theorem:triangleRemoval}, it follows that the number of edges that need to be removed to make the graph triangle-free is at most $o(|V|^2)=o(N^2)$. Since this naturally also bounds the total number of edges $|E|$, which is at least $3Nn/4$, we obtain that  
$3Nn/4\leq o(|V|^2)=o(N^2)$ and 
hence $n=o(N)$. But from \Cref{lemma:size_bound}, $N$ is bounded by a constant multiple of $n$, hence we get that $n=o(N)=o(n)$ which is a contradiction. The main technical challenge of this proof is in constructing a map between the set $|P+Q|$ and vertices of the graph in such a way so that the condition of Property~\ref{property:decodability}, that $p_k+q_k =p_i+q_j $ if and only if $i=k,j=k$, translates into the condition that every edge is in exactly one triangle. Therefore, the first part of this proof will be about defining a map $\psi$ with the desired properties, as presented in the proof of~\cite[Theorem 3.5, Chapter 2]{Ruzsa2009}.

\begin{proof}[Proof of \Cref{theorem:lower}]
Assume for contradiction that $|P+Q|=K|P|$. Let us call $N=|P+Q-P-Q|$. 

First, we construct $\phi:\mathbb{Z}\rightarrow [q]$ for a prime  $q>\max(P+Q-P-Q)$, such that if $x\in P+Q-P-Q$ and $\phi(x)=0$, then $x=0$ (we do not simply define $\phi(x)=x\mod q$, since we need $\phi$ to satisfy additional properties). We do this in the following way.

For each $\lambda\in[q]$, consider the map 
$$ \phi_{\lambda}: \mathbb{Z} \xrightarrow{\mathrm{mod}\ q} \mathbb{Z}/q\mathbb{Z} \xrightarrow{\times \lambda} \mathbb{Z}/q\mathbb{Z}\xrightarrow{(\mathrm{mod}\ q)^{-1}}[q].$$
It is important to note that (mod $q)$ and $\times \lambda$ are group homomorphisms. Therefore, as the operation mod $q$ restricted to $[q]$ is the identity map, $\phi_\lambda$ has the following property~\cite{Ruzsa2009}:

\begin{property}
\label{property:sum}
If $a_i,x_i\in \mathbb{Z}$ and $0 \leq \sum_{i=1}^n a_i \phi_{\lambda}(x_i)<q$ then $\sum_{i=1}^n a_i \phi_{\lambda}(x_i)=\phi_{\lambda}(\sum_{i=1}^n a_i x_i)$.
\end{property}

For each $x\neq 0$, $\phi_\lambda (x)$ takes on the values of $[q]$ with equal probability over all $\lambda\in [q]$. Therefore, for each $x\in (P+Q-P-Q)\setminus \{0\}$, the probability over all $\lambda \in [q]$ that $N$ divides $\phi_{\lambda}(x)$ is at most $1/N$. As there are $N-1$ elements in $(P+Q-P-Q)\setminus \{0\}$, by the union bound, we get:
\begin{eqnarray*}
  &&\Pr_{\lambda\in [q]}\left(\forall x\in (P+Q-P-Q)\setminus \{0\}, N \text{ does not divide } \phi_\lambda (x)\right)\\
  &\ge& 1- \left(\sum_{x\in (P+Q-P-Q)\setminus \{0\} }\Pr_{\lambda \in [q]}\left(  \text{$N$ divides }\phi_\lambda (x)\right)\right)\\
  &\ge& 1- \frac{N-1}{N}=1/N>0. 
\end{eqnarray*}

This implies that there exists $\lambda\in [q]$ such that for all $x\in P+Q-P-Q\setminus\{0\}$ it holds that $N$ does not divide $\phi_{\lambda} (x)$. We fix $\lambda$ to be such a constant, and set $\phi=\phi_{\lambda}$. 
Now, we construct a set $T'\subseteq [n]$ so that we have the following property 
\begin{property}
\label{property:set}
For all $i,j,k\in T'$, we have:
$|\phi(p_k)-\phi(p_i)+\phi(q_k)-\phi(q_j)|\le q$.
\end{property}
In order to construct $T'$, we first pick one of the two intervals $[0,\frac{q-1}{2}]$ and $[\frac{q+1}{2},q-1]$ to be $I$ so that $T=\{ k\in \{1\ldots n\}:\phi(p_k)\in I\}$ and $|T|\ge n/2$. We then pick an interval $J$ to be one of  $[0,\frac{q-1}{2}]$ and $[\frac{q+1}{2},q-1]$ so that $T'=\{k\in T:\phi(q_k)\in J$\} and $|T'|\ge |T|/2\ge n/4$.
We now compose $\phi$ with (mod $N$) to obtain a mapping $\psi$, as follows.
    $$ \psi: \mathbb{Z} \xrightarrow{\phi} \{0,1\ldots q-1\} \xrightarrow{\text{mod N}} \mathbb{Z}/N\mathbb{Z}.$$

Now, we construct a tripartite graph $G$ whose vertex set is $V=X\dot{\cup} Y\dot{\cup} Z$,where $X,Y,Z$ are all copies of $\mathbb{Z}/N\mathbb{Z}$. The edges of the graph are defined as follows.
\begin{itemize}
    \item $(x,y)\in X\times Y$ is in $E(G)$ if and only if there exists $i\in T'$ such that $y-x= \psi(p_i)$.
    \item $(y,z)\in Y\times Z$ is in $E(G)$ if and only if there exists $i\in T'$ such that $z-y= \psi(q_i)$.
    \item $(x,z)\in X\times Z$ is in $E(G)$ if and only if there exists $i\in T'$ such that $z-x=\psi(p_i)+\psi(q_i)$ mod $N$.
\end{itemize}
Note that $|V(G)|=|X|+|Y|+|Z|=3N$, and $|E(G)|=3N|T'|\ge \frac{3}{4}Nn$. If $x,y,z$ form a triangle, then there exist $i,j,k$ such that $y-x=\psi(p_i)$, $z-y=\psi(q_j)$, and $z-x=\psi(p_k)+\psi(q_k)$. We claim that in this case it must hold that $i=j=k$, which we prove as follows. Since $(z-y)+(y-x)=z-x$, we have that $\psi(p_k)+\psi(q_k)=\psi(p_i)+\psi(q_j)$, which means that $\psi(p_k)+\psi(q_k)-\psi(p_i)-\psi(q_j)=0$. Since $\psi$ applies a (mod $N$) function to $\phi$, this means that $N$ divides $\phi(p_k)+\phi(q_k)-\phi(p_i)-\phi(q_j)$. By Property~\ref{property:set} and our choice of $T'$, it holds that $\phi(p_k)+\phi(q_k)-\phi(p_i)-\phi(q_j) \leq q$, and hence by Property~\ref{property:sum}, we have that $\phi(p_k)+\phi(q_k)-\phi(p_i)-\phi(q_j) = \phi(p_k+q_k-p_i-q_j)$. This implies that $N$ divides $\phi(p_k+q_k-p_i-q_j)$, which only holds if $p_k+q_k-p_i-q_j = 0$ due to our choice of $\lambda$.

Since such a triangle implies $i=j=k$, this shows that each edge is in exactly one triangle hence, and so $Nn/4$ edges need to be removed to make $G$ triangle free. The number of triangles is $Nn/4=O(N^2)=o(N^3)$ and therefore, by the Triangle Removal Lemma from \Cref{theorem:triangleRemoval}, one can remove at most $o(N^2)$ edges to make $G$ triangle free. This implies that $Nn/4=o(N^2)$, and hence $n=o(N)$. However, we proved in \Cref{lemma:size_bound} that $N\le K^7n$, and hence $n\neq o(N)$ which is a contradiction. Therefore, our initial assumption does not hold, i.e., there is no constant $K$ such that $|P+Q|\le K|P|$ for all $n=|P|$. This implies that $|P+Q|=\omega(n)$ as claimed.
\end{proof}

\section{An $\nummats e^{O(\sqrt{\log{n}})}$ Upper Bound for the Recovery Threshold}
\label{sec:UB}

In this section, we show how to obtain Rook Codes for batch matrix multiplication in the master-worker setting with a near-optimal recovery threshold.
\ThmUB*

To prove~\Cref{thm:y}, we show that there exist sets $P,Q$ with $|P+Q|=ne^{O(\sqrt{\log n})}$ such that the decodability property (Property~\ref{property:decodability}) holds. 

In order to show this, we use the following known construction of a set which does not contain any 3-term arithmetic progression, defined as follows.
\begin{definition}
\label{def:AP}
A set $A=\{a_1,\dots,a_\ell\}$ is an $\ell$-term arithmetic progression ($\ell$-AP) if there is a value $d$ such that $a_i = a_1 + (i-1)d$, for all $2\leq i \leq \ell$.
\end{definition}

Behrend~\cite{Behrend_1946} showed a construction of a set which does not contain any 3-AP, whose maximum element is bounded by a value which will be useful for us.

\begin{theorem}[Behrend~\cite{Behrend_1946}]
\label{thm:behrend}
    For all $n$, there exist a set $A$ of distinct positive integers of size $|A|=n$ which does not contain any 3-term arithmetic progressions (3-AP), and such that $\max_{a\in A} \{a\} = n e^{O(\sqrt{\log n})}$.
\end{theorem}

Note that the proof of~\Cref{thm:behrend} 
shows that there exists a 3-AP-free subset $A$ of $\{1,2,\ldots N \}$ of size at least $N^{1-\frac{2\sqrt{2\log 2} +\epsilon}{\sqrt{\log N}}}$ elements. By choosing $n=N^{1-\frac{2\sqrt{2\log 2} +\epsilon}{\sqrt{\log N}}}$,  we have that $\log{n}=\Theta(\log N)$. Therefore, expressing $N$ in terms of $n$ we get that $N=ne^{O(\sqrt{\log n})}$, hence for any $n$, we can construct a 3-AP-free set with $n$ elements such that $\max_{a\in A} \{a\} = n e^{O(\sqrt{\log n})}$, which is the form of the theorem presented here. With this construction, we can easily prove our upper bound on the recovery threshold of Rook Codes, as follows.

\begin{proof}[Proof of~\Cref{thm:y}]
Let $c$ be a suitable constant such that for each $n$, there is a set $A_n\subseteq [ne^{c\sqrt{\log n}}]$ of size $n$ which contains no 3-AP, as obtained by Behrend's construction in~\Cref{thm:behrend}. Denote $A_n=\{a_1,\ldots , a_n \}$.
Let $P,Q=A$ where $p_i=q_i=a_i$. This choice of $P$ and $Q$ satisfies the decodability property (Property~\ref{property:decodability}), because if $2p_i=p_k+q_j$ then $p_k,p_i,q_j$ is a 3-AP, hence $p_k,p_i,q_j$ must all be equal. Furthermore, $P+Q\subseteq \{1,2\ldots 2ne^{c\sqrt{\log n}}\}$, with a size of $ne^{O(\sqrt{\log n})}$. Therefore, $\Recovery_{\text{Rook-Codes}}(n)=ne^{O(\sqrt{\log n})}.$
\end{proof}

\section{The Computational Efficiency of Rook Codes }\label{appx:parr_and_comp}

We provide here the mathematical background that could explain the experimental results of~\cite{sfsl22}, which suggest that Rook Codes may have faster computational times. 

~\\\textbf{\underline{Coding for batch matrix multiplication.}} In~\cite{ylrksa19}, the method of Lagrange Coding Computation (LCC) has been suggested for batch matrix multiplication. 
In this approach, the master defines two polynomials $\tilde{A}(x)$ and $\tilde{B}(x)$, by invoking point evaluation over the input matrices. In more detail, $\tilde{A}(x)$ is defined to be the unique polynomial of degree at most $n-1$, which is interpolated from $n$ points $z_0,\dots,z_{n-1}$ whose values are $A_0,\dots,A_{n-1}$, respectively. That is, $\tilde{A}(z_i)=A_i$ for all $0\leq i \leq n-1$. Similarly, $\tilde{B}(x)$ is interpolated so that $\tilde{B}(z_i)=B_i$ for all $0\leq i \leq n-1$. Then, for $0\leq w \leq m-1$, the master sends to worker $w$ the coded matrices $\tilde{A}(x_w)$ and $\tilde{B}(x_w)$, for some value $x_w$. Each worker $w$ multiplies its two matrices and sends their product back to the master. Note that $(\tilde{A}\cdot\tilde{B})(x)=\tilde{A}(x)\cdot\tilde{B}(x)$ is a polynomial of degree at most $2n-2$, for which $(\tilde{A}\cdot\tilde{B})(z_i)=A_i\cdot B_i$ for all $0\leq i \leq n-1$. Thus, once the master receives back $\tilde{A}(x_w)\cdot\tilde{B}(x_w)$ from $2n-1$ workers, it can extract $A_i\cdot B_i$ for all $0\leq i \leq n-1$ by interpolating the polynomial $(\tilde{A}\cdot\tilde{B})(x)$ and evaluating it at $z_0,\dots,z_{n-1}$. This yields $\Recovery_{\text{Lagrange-Codes}}(n)\leq 2n-1$.

In~\cite{jj21}, a method of Cross Subspace Alignment (CSA) Coding for batch matrix multiplication is suggested, in which rational functions are used. 
More precisely, for a vector $\Vec{z} = (z_0,...,z_{\nummats-1})$ of $\nummats$ distinct values in a field, let
$f_{\Vec{z}} (x) = \prod_{i \in[\nummats]} ( z_i - x) = ( z_0 - x)\cdot ... \cdot ( z_{\nummats-1} - x)$, and define the rational encoding functions as $\tilde A (x ) = f_{\Vec{z}} (x)  \left(\sum_{i \in [\nummats]} \frac{1}{z_i - x}A_i  \right) $,
and
$\tilde B (x ) =  \sum_{i \in [\nummats]} \frac{1}{z_i - x}B_i  $.
Then it holds that 
$(\tilde A \cdot \tilde B) (x )  = \tilde A (x) \tilde B (x ) = \sum_{i \in [\nummats]} \frac{c_i(\vec{z})}{z_i - x}A_iB_i + \sum_{i \in [\nummats-1]} x^i N_i$,
where each $c_i(\vec{z})$ does not depend on $x$  
and each $N_i$ denotes a noise (matrix) coefficient that can be disregarded. The master computes 
$\tilde A (x_w)$ and $\tilde B(x_w)$ for some distinct $x_w$ that satisfy $\{z_i\}_{i \in [n]} \cap \{x_w\}_{w\in [m]} = \emptyset $ 
and sends it to worker $w$, for each $0\leq w\leq m-1$. Each worker $w$ then multiplies its two matrices and sends their product back to the master. Note that $(\tilde{A}\cdot\tilde{B})(x)=\tilde{A}(x)\tilde{B}(x)$ is a rational function with at most $2n-1$ zeros and $n$ poles\footnote{For any $ p=2n-1$ points $\{(x_i,y_i)\}_{i\in[p]}$ and $n$ numbers $\{(z_i)\}_{i\in[n]}$ such that $\{z_i\}_{i \in [v]} \cap \{x_k\}_{k\in [p]} = \emptyset $, there exists a unique rational function $h (x)= f(x)/g(x) $ such that $h(x_i) = y_i $, $g(z_i) = 0 $, $\mathrm{deg}(f) = 2n-1$, and $\mathrm{deg}(g) = n$. Notice that this is a natural generalization of polynomial interpolation (see, e.g., ~\cite[Lemma~1]{jj21}).}, 
and so the master can perform rational interpolation from the values sent to it from any $2n-1$ workers. Then the master retrieves $A_i\cdot B_i$ for all $0\leq i \leq n-1$, by taking the coefficient of $c_i(\vec{z})/z_i-x$.  
This yields $\Recovery_{\text{CSA-Codes}}(n)\leq 2n-1$.

~\\\textbf{\underline{A comparison with Rook Codes.}} 
As mentioned, at a first glance, it may seem that LCC or CSA Codes already achieve the desired linear recovery threshold. However, Rook Codes are powerful in their computational complexity, which motivates our work on bounding their recovery threshold. To see this, let us be more specific about the parameters of the setting, as follows. For all $0\leq i\leq n-1$, denote by $\chi\times\zeta $ the dimensions of the matrix $A_i$ and by $\zeta\times\upsilon$ the dimensions of the matrix $B_i$. 

\textbf{\underline{Encoding.}} For the encoding time, all three methods  (Rook Codes, LCC, and CSA Codes) must use some fast multipoint evaluation algorithm. In particular, Rook Codes and LCC use multipoint evaluation of a polynomial, and CSA Codes use a multipoint evaluation of rational functions, i.e., Cauchy matrix multiplication \cite{vg2013}.
All three methods can compute the $m$ values $\{ \tilde A(x_w), \tilde B(x_w) \}_{w \in [m]}$ in time
$O((\chi \zeta + \zeta\upsilon) \max\{ m,d\}\log^c\max\{ m,d\}) = O((\chi  + \upsilon)\zeta m\log^cm)$, where $c$ is some constant and $d$ is the degree of the polynomials $\Tilde{A}(x),\Tilde{B}(x)$ in Rook Codes and LCC, or the total number of their roots and poles in CSA~\cite{vg2013}.  

\textbf{\underline{Decoding in Rook Codes.}} For the decoding time, in Rook Codes, the master can decode in time $O(\chi  \upsilon R\log^c(R))$, for interpolating a polynomial over $\chi \times \upsilon$ matrices out of $R= \Recovery_{\text{Rook-Codes}}(n) $ points. 
This gives that the total computational time (for encoding and decoding) at the master is $O ( (\chi  + \upsilon)\zeta m \log^c (m) + \chi  \upsilon R\log^c(R))
$.
This decoding time is much smaller than the encoding time in some cases, for example when $\chi=\upsilon$ and $\zeta = \omega(\chi R\log^c{R}/(m\log^c{m}))$.
Since in such cases the encoding is the more expensive procedure, this motivates delegating the encoding to the workers in scenarios in which the master should do as little computation as possible (for example, when it is responsible for more than a single task). 
This is done by sending all matrices $A_i$ and $B_i$ for $0\leq i \leq n-1$ to all workers, and having each worker $w$ compute $A(x_w), B(x_w)$ (the values in $P,Q$ can be either also sent by the master or be agreed upon as part of the algorithm, depending on the setup).
While this clearly has an overhead in terms of communication, in some cases this overhead can be very small compared to the original schemes in which encoding takes place at the master -- as an example, if all workers are fully connected then the master can send $A_i, B_i$ to a different worker $w_i$ for each $i$, and each $w_i$ can forward these two matrices to all other workers concurrently.

The benefit of delegating the encoding to the workers is that it allows the master to avoid any computation for encoding, and each worker $w$ only has to evaluate the polynomial at its own $x_w$ value, which needs at most $(\chi+\upsilon)\zeta n + O(n \log^{1/2}(n))$ multiplications (see Section~\ref{appx:rook}). 
The $ O(n \log^{1/2}(n))$ term does not depend on $\chi, \upsilon, \zeta$ and thus, is negligible when the matrices have large dimensions.

\textbf{\underline{Decoding in LCC and LCA.}} The crucial difference between Rook Codes and LCC or CSA Codes, is that in Rook Codes, worker $w$ can locally encode $\tilde{A}(x_w),\tilde{B}(x_w)$ using division-free algorithm (i.e., an algorithm using only $+$ or $\cdot$, but not $/$); this stands to LCC and CSA Codes, which must perform divisions (by definition).
Furthermore, it is not known how to efficiently delegate the encoding procedure to the workers in LCC since the encoding functions are 
$\tilde A(x) = \sum_{i \in n} A_{i}\prod_{j \in [n] \setminus \{i\}}\frac{x - z_j}{z_i - z_j}$ and $\tilde B(x) = \sum_{i \in n} B_{i}\prod_{j \in [n] \setminus \{i\}}\frac{x - z_j}{z_i - z_j}$ and, in particular, since there are $O(n^2)$ many products of the form $1/(z_i- z_j)$, any na\"ive algorithm must take quadratic time, and converting it to coefficient form for using Horner's scheme (see Section~\ref{appx:parr_and_comp}) incurs a $O(n\log^c n )$ overhead. 
Thus, the encoding step for LCC has complexity $O((\chi  + \upsilon)\zeta n\log^c n )$, by~\cite[Theorem 10.10]{vg2013}. 
Moreover, divisions are known to be more expensive (for both LCC and CSA)
\cite{h03,vg2013}, thus they increase the overhead at each worker. In particular, the current theoretical bound (using Newton iterations) on division in terms of multiplication for number fields is 3 \cite{h03}.
Notice that evaluating the encoding function $\tilde B(x) = \sum_{i\in [n]} B_i \frac{1}{z_i-z} $ at a point $x_w$ is very unlikely to use less than $n(\chi + \upsilon) \zeta $ divisions since even if we give it a common denominator, i.e., $\tilde B(x) = (\prod_{i\in [n]} \frac{1}{z_i-z})\sum_{i\in [n]} B_i(\prod_{i\in [n]} {z_i-z})  \frac{1}{z_i-z} $, we would now have a super-linear increase in multiplications in the numerator. Furthermore, the evaluation of $\tilde B (x_w)$ is equivalent to computations involving Cauchy matrices, an old computational problem that has been studied extensively and conjectured to have a super-linear lower bound in the number of points plus poles (see \cite{vg2013}).
Furthermore, division-free algorithms are often preferred because they tend to avoid issues of division-by-0 and false-equality-with-0 \cite{r01}. 

~\\To summarize, there are settings in which Rook Codes are preferable over LCC and LCA codes, which motivates bounding their recovery threshold to obtain high robustness to faults.

\subsection{Bounding the Encoding Time of Our Construction of Rook Codes }
\label{appx:rook}

We provide here the promised analysis of the encoding time of our construction.

\begin{lemma}
\label{lemma:rook_comp}
   Let $p_0<p_1<...< p_{n-1}$ and $q_0<q_1<...< q_{n-1}$ be increasing orderings of the elements in $P$ and $Q$, respectively, and let $\delta(P,Q) $ equal the number of multiplications needed to compute the values $\{x^{p_i - p_{i-1}}\} _{i \in [n]} \cup \{x^{q_i - q_{i-1}}  \} _{i \in [n]}$, where $p_{-1}=q_{-1}=0$. Then, given $A_i,B_i$ for all $i$, each worker $w$ can locally encode its values $\Tilde{A}(x_w)$ and $\Tilde{B}(x_w)$ using at most $\delta(P, Q) + (\chi  + \upsilon)\zeta n $  multiplications.
\end{lemma}

\begin{proof}
Each worker $w$ uses Horner's rule~\cite{CLRS} to compute $\tilde A(x_w), \tilde B(x_w)$. There are at most $n$ non-zero coefficients in $\tilde A(x), \tilde B(x)$, and thus Horner's rule takes the form
$\Tilde{A} (x) = x^{p_0}(A_0 + x^{p_1-p_0}(A_1 + x^{p_2-p_1}(A_2+...)))$ for $\tilde{A}(x)$ (and similarly for $\tilde{B}(x))$, where the values $c_{w,p,k} = x_w^{p_k- p_{k-1}}$ and $c_{w,q,k} = x_w^{q_k- q_{k-1}}$ can all be computed using successive squaring for at most a $\delta(P,Q)$ number of multiplications. More precisely, we have that $f_{n-1}(x_w) := \Tilde{A}(x_w)$ can be recursively computed as
$f_{i+1}(x_w) := c_{w,p,n-1-i}(A_{n-1-i} + f_{i}(x_w))$,    where $f_0 = x_w^{p_{n-1}-p_{n-2}}A_{n-1}$. Therefore, computing $\tilde A(x_w)$ can be done using at most $\delta(P, Q) + (\chi  + \upsilon)\zeta n $  multiplications and $(\chi  + \upsilon)\zeta n $ additions. A similar argument applies for computing $\tilde B(x_w)$.
\end{proof}

The following lemma bounds $\delta(P,Q)$, so that plugging this bound in Lemma~\ref{lemma:rook_comp} gives our claimed time for encoding of our Rook Codes.
\begin{lemma}
$\delta(P,Q) = O(n \log^{1/2}(n) )$.
\end{lemma}

\begin{proof}
Consider our Rook codes of Theorem~\ref{thm:y} that use Behrend's construction~\cite{Behrend_1946}. For some constants $C_0,\dots,C_4$, define $d,\ell$ and $k$ as 
$$ \frac{C_0\log(n)}{2\ell}-\frac{1}{2} <  d \leq \frac{C_1\log(n)}{2\ell}+\frac{1}{2} ,
$$
$$ C_2 \log^{1/2}(n) \leq \ell = \sqrt{\frac{2 \log (N)}{\log (2) }}  < C_3 \log^{1/2}(n) ,
$$
and 
$$k \leq \ell (d-1)^2 < C_4 \log^{3/2}(n) .$$
Then we have that  
$$
P = Q = 
\{ a_0(2d-1)^0 + ... +  a_{\ell -1}(2d-1)^{\ell-1}
\ \vert \
\sum_{i \in [\ell]} a_i^2 = k 
,\  a_i \leq d
\}.
$$

It takes at most $O(\log (d)) $ many steps to compute $x^{2d-1}$ for some value of $x$.
Therefore, by the definition of $P$, the number of multiplications can be counted in terms of multiplying powers of $x$.  
In particular, if $p_j - p_{j-1} = \sum_{i \in [\ell] } b_{i}^{(j)} (2d-1)^i$ then we have that 
\begin{equation}\label{eq:s_sq_alg}
x^{p_j - p_{j-1} } = x^{b_0}x^{\sum_{i \in [1,\ell] } (2d-1)^ib_{i}^{(j)}} = x^{b_0}\prod_{i \in [1,\ell] }  (x^{(2d-1)^i})^{b_{i}^{(j)}}. 
\end{equation}
It takes at most $ O(\ell\log(d)) $ multiplications to compute $x,x^{2d-1},...,x^{(2d-1)^{\ell-1}}$ successively (since computing $x^{(2d-1)^{i+1}}$ given $x^{(2d-1)^{i}}$ only takes $O(\log(d))$ operations using successive squaring) and it takes at most $O(\ell d)$ multiplications to compute all of the possible values $(x^{(2d-1)^i})^{b_{i}^{j}}$ since the $b^j_i$ values are bounded by $O(d)$; therefore, each value in Equation~\eqref{eq:s_sq_alg} takes at most $\ell$ many multiplications to compute (after the first three preprocessing steps). 
We conclude that the number of multiplications needed is bounded by $O(n \ell ) =O(n \log^{1/2}(n) )  $ since Equation~\eqref{eq:s_sq_alg} is computed at most $n$ times.
\end{proof}

\section{Discussion}
This paper analyzes batch matrix multiplication in the master-worker setting that is both efficient and robust, by showing nearly matching upper and lower bounds for the recovery threshold of Rook Codes. 

Our lower bound in Theorem~\ref{thm:x} is surprising since, in the classical case of polynomial codes \cite{hp03, v98} (or Reed-Solomon codes \cite{rs60}), the communication rate is the same whether or not one defines the codewords as evaluations or as the coefficients of a polynomial. 
However, in the case of Rook codes and LCC we have that the recovery threshold (the measure of complexity that we care about here for batch matrix multiplication) is different between the evaluation codewords (Rook coding) and coefficient codewords (LCC). 

~\\We state here some further research directions.

\textbf{Constructing 3-AP free sets.} In the above setting, it is assumed that the number of pairs of matrices $n$ is fixed. One may consider the non-uniform case, in which $n$ is given as a parameter. In such a setting, the master would have to construct the appropriate 3-AP free sets according to $n$ as part of its computation. For this, one would need a construction that is also computationally efficient. It may be an interesting direction for further research and, in particular, the construction of \cite{Elkin10} could turn out useful. The latter constructions is also better in lower order terms of its size.

\textbf{Possible improvements of the lower bound.}
In additive combinatorics, problems such as finding 3-AP free sets are known as tri-colored sum-free problems \cite{Kleinberg2018,Aaronson2016,fox2017}. Although these are similar to 3-AP free sets, the bound achieved is stronger for 3-AP free sets than for tri-colored sum free sets. If $A\subset [N]$ is a 3-AP free set, it is proven that $|A|= O(\max(A)\exp{(-\Omega((\log N)^{\beta}))})$ for some constant $\beta>0$ ~\cite{kelley2023}. However, such strong qualitative bounds are not yet found for tri-colored sum free sets. In the language of \cite{zhao2023}, one reason for this is that the equality $a+b=2c$, which is characteristic of the 3-AP problem, is preserved even after $a,b,c$ is translated, whereas the equality $p+q=r$, which is characteristic of the tri-colored sum-free set is not preserved after translation (i.e., $(p+t)+(q+t)\neq (r+t)$). The proof for 3-AP free sets exploits this property by iteratively selecting sub-sequences of $A$, then scaling and translating this set to increase the density $|A|/\max\{A\}$ of the set. It is an open question whether the lower bound can indeed be increased to exactly match our upper bound.

\section*{Acknowledgements}
This project has received funding from the
European Union’s Horizon 2020 research and innovation programme under grant agreement no. 755839, and from ISF grant 529/23. We thank Yufei Zhao for his advice for an alternate and more concise proof of~\Cref{theorem:lower}.

\bibliography{skeleton.bib}

\end{document}